\documentclass[12pt]{amsart}
\usepackage{amsmath, amsthm, amsfonts, amsbsy, amssymb, upref, enumerate, bigstrut,  color, mathtools, mathrsfs, float, bm}

\usepackage[left=1.3in,top=1.3in,right=1.3in]{geometry}

\usepackage{epsfig, graphicx}
\usepackage{ hyperref}

\newtheorem{theorem}{Theorem}[section]
\newtheorem{lemma}[theorem]{Lemma}
\newtheorem{corollary}[theorem]{Corollary}

\newtheorem{example}[theorem]{Example}

\begin{document}
\title[Phase Transition in Random Polymers]
{Existence of Scaling Limit Phase Transition in a Two-dimensional Random Polymer Model
}
\author{Luis R. Lucinger}
\address{
\newline
Departamento de Matem\'atica -
Universidade de Bras\'{i}lia, Brazil,
Email: \textup{\tt lucinger@mat.unb.br}
}
%
%
\author{R. Vila$^\dag$}
\address{
\newline 
Departamento de Estat\'istica -
Universidade de Bras\'ilia, Brazil,
Email: \textup{\tt rovig161@gmail.com}
}

\date{\today}


\keywords{Random polymer; Positive Association.
\\	
$^\dag$ \  Corresponding author.}

\begin{abstract}
In this paper we prove a scaling limit phase transition for a class of two-dimensional random polymers.
\end{abstract}

\maketitle
\vspace*{-0,2cm}
\section{Introduction}\label{Introduction}
In this paper we consider two-dimensional random polymers, 
which are defined as follows. 
Given $N\in\mathbb{N}$, a $N$-th step polymer $\mathcal{S}$
is an element of $\mathbb{W}_N$ given by  
\[
\mathbb{W}_N\coloneqq \big\{
\mathcal{S}=(\mathcal{S}_0,\mathcal{S}_1,\dots, \mathcal{S}_N) : \, 
\mathcal{S}_i\in \mathbb{Z}^2, \, \mathcal{S}_0=0 \ 
\text{and}\ \|\mathcal{S}_{i+1}-\mathcal{S}_i\|_1=1\big\},
\]
where $\|\cdot\|_1$ denotes the taxicab norm (or Manhattan distance). 
Its probability distribution is defined by a Gibbs measure,
at the inverse temperature $\beta>0$, given by
\begin{align}\label{model}
	\mathbb{P}_{N}^{\beta}(\{\mathcal{S}\})
	=
	{1\over \mathcal{Z}_N(\beta)} \exp\big(-\beta\mathcal{H}_N(\mathcal{S})\big),
\end{align}
where
$
\mathcal{Z}_N(\beta)=\sum_{\mathcal{S}} \exp\big(-\beta\mathcal{H}_N(\mathcal{S})\big)
$
is a normalization factor and the Hamiltonian is
\begin{align}\label{hamiltonian}
	\mathcal{H}_N(\mathcal{S})
	=
	{\sum_{1\leqslant i<j\leqslant N}V_{ij}\cdot \langle \mathcal{X}_i,\mathcal{X}_j\rangle},
\end{align}
with $\mathcal{X}_i\coloneqq\mathcal{S}_{i}-\mathcal{S}_{i-1}$,
$V_{ij}\geqslant 0$ 
and $\langle\cdot,\cdot\rangle$ the usual inner product.
Here we assume that the interaction $V_{ij}$ depends only on the distance between 
$i, j\in\mathbb{N}$, i.e., $V_{ij} = V(|i-j|)$ and satisfies the following regularity condition
\[
\sum_{i\in\mathbb{N}}V(i)<+\infty.
\]
\quad 
Such random polymers have received considerable attention in the literature 
bringing together physics, chemistry, and more recently biophysics.
In Caracciolo et al. \cite{MR1301458}, the authors considered this type of a random polymer model, where the two body interactions decay as a power-law.
Their model interpolates between the lattice Edwards model and an ordinary simple random walk (SRW).
Under the assumption of $1 \leqslant d \leqslant 4$, the authors proved
the existence of an exponent $\gamma=\gamma(d,\alpha)$, for the end-to-end distance, such that $\sum_{\mathcal{S}}\|\mathcal{S}_N\|^2\mathbb{P}^\beta_N(\mathcal{S})\sim cN^{2\gamma}$ 
holds.
In Butt\`a et al. \cite{MR2151217}, by considering an appropriate scale, the authors showed that
the end-to-end distance of a two-dimensional random polymer with
a self repelling interaction of Kac-type undergoes a diffusive-ballistic phase transition. 
In Procacci et al. \cite{MR2379704}, the authors considered a self repelling model
with the Hamiltonian given by 
$
\sum_{1\leqslant i<j\leqslant N}\widetilde{V}_{ij}
\cdot \langle \mathcal{S}_i-\mathcal{S}_j,\mathcal{S}_i-\mathcal{S}_j\rangle,
$
where the
interactions $\widetilde{V}_{ij}=1/|i-j|^\alpha$, with $3<\alpha\leqslant 4$.
They proved that their model is diffusive at sufficiently high temperatures and ballistic at sufficiently  low temperatures.
Random polymer models with some similarity to the model
proposed in Procacci et al. \cite{MR2379704} have been studied in 
Bouchaud et al. \cite{MR3456987};
Caracciolo et al.  \cite{MR1301458};
van der Hofstad et al. \cite{MR1633582}; and
Marinari and Parisi \cite{MR1234567}.
Recently, Cioletti et al. \cite{MR3226841} 
considered a similar random polymer model, where the Hamiltonian is given by
\eqref{hamiltonian} and interactions are of the form $V_{ij}=1/|i-j|^\alpha$, $1<\alpha\leqslant 2$. 
The authors proved the existence of a ballistic-diffusive phase transition in terms of the inverse temperature $\beta$, by comparing their long-range polymer model with two 
independent copies of a long-range one-dimensional ferromagnetic Ising model. 
In addition, they obtained a Central Limit Theorem (CLT) 
for a similar model with drift.

Let us give an informal explanation 
why we called the polymer models in \cite{MR3226841} and \cite{MR2379704} similar, 
although the mentioned phase transitions, require $\alpha$ to be in different intervals. 
The idea is to look 
at these random polymer models as spin models in the lattice $\mathbb{N}$.
In \cite{MR3226841} the spin variables are $(\mathcal{X}_i)_{i\in\mathbb{N}}$, and
in \cite{MR2379704} the spin variables are $(\mathcal{S}_i)_{i\in\mathbb{N}}$.
Rewriting the model in \cite{MR2379704}  by using    
spin variables $(\mathcal{X}_i)_{i\in\mathbb{N}}$,
we get for configurations which are ``small'' 
perturbations of a ground state that
\begin{align*}
	\sum_{1\leqslant i<j\leqslant N}
	\widetilde{V}_{ij}
	\cdot \langle \mathcal{S}_i-\mathcal{S}_j,\mathcal{S}_i-\mathcal{S}_j\rangle
	&=
	\sum_{1\leqslant i<j\leqslant N}
	\widetilde{V}_{ij}
	\sum_{k=i+1}^{j} \sum_{l=i+1}^{j}\langle \mathcal{X}_k,\mathcal{X}_l \rangle	
	\\[0.2cm]
	&\sim 
	\sum_{1\leqslant i<j\leqslant N}
	\!\!\!\widetilde{V}_{ij} \, |i-j|^2  
	\cdot \langle \mathcal{X}_i,\mathcal{X}_j\rangle,
\end{align*}
which explains why $\alpha\in (1,2]$ in \cite{MR3226841} and 
$\alpha\in (3,4]$ in \cite{MR2379704}.

In the present paper we prove a phase transition 
in terms of suitable scaling limits.
We consider a class of power-law decaying interactions
as in \cite{MR3226841}, 
and we prove that the critical inverse temperature where the 
polymer model behavior changes is equal to $\beta_c$, the critical inverse
temperature  of a related one-dimensional Ising model.
Roughly speaking, we prove the existence of a critical point $\beta_c\in(0,+\infty)$
such that the scaling limit of our random polymer model converges in the Wasserstein  
distance 
(see the beginning of the next section for its definition) 
to the planar standard Brownian 
motion in the subcritical regime 
$\beta<\beta_c$.
Moreover, we also prove that it does not scale to the Brownian motion when $\beta>\beta_c$.
To be more precise, the statements of the main result of this paper are the following.
Let $\mathcal{S}\in {\mathbb W}_N$ be a random polymer
and consider the stochastic process
\begin{align*}
	\mathcal{W}_n(t)\coloneqq
	\frac{1}{\sigma\sqrt{n}}\big\{\mathcal{S}_k+(nt-k)(\mathcal{S}_{k+1}-\mathcal{S}_{k})\big\},
	\quad\frac{k}{n}\leqslant t<\frac{k+1}{n}, \ t\in[0,1].
\end{align*}

\begin{theorem}[Phase transition] \label{Donsker invariance principle}
	Consider the two-dimensional random polymer model defined by the Gibbs measure
	\eqref{model}-\eqref{hamiltonian},
	where the interaction $V_{ij}=|i-j|^{-\alpha}$,
	and $1<\alpha\leqslant 2$ is fixed.
	For each $0<p\leqslant 2$ the following holds:
	\begin{enumerate}
		\item
		if $\beta<\beta_c$, 
		then $\lim\limits_{n\to\infty}d_p(\mathcal{W}_n(t),\boldsymbol{B}_2(t))=0$;
		\item
		if $\beta>\beta_c$, 
		then $\liminf\limits_{n\to\infty}d_p(\mathcal{W}_n(t),\boldsymbol{B}_2(t))>0$;
	\end{enumerate}
	where $\boldsymbol{B}_2(t)=(B_t^1, B_t^2)$ is the 
	planar standard Brownian motion.
\end{theorem}

We emphasize that the convergence obtained in Theorem \ref{Donsker invariance principle} is stronger than the convergence in distribution. Moreover, by a result of Bickel and Freedman \cite{MR630103}, it also implies convergence of moments.

This paper is organized as follows.
In Section \ref{sec-war} we introduce the Wasserstein distances and recall some of 
their basic results.
Later, in Section \ref{sec-conv}, we prove an abstract theorem about the scaling limits of positively correlated random variables.
We first consider one-dimensional processes and provide an application on the long-range Ising Model.
Next we consider $m$-dimensional processes and use this version to prove our main result in the last section.

\section{Wasserstein Distance}\label{sec-war}

The Wasserstein distance \cite{MR0314115} is
also known as Monge-Kantorovich-Rubinstein distance \cite{MR0102006}, 
Mallows distance \cite{MR0298812} or
optimal transport distance in optimization \cite{MR2011032}. 
It is, among other things, a useful tool in order to derive CLT type results including
the case of heavy-tailed stable distributions 
(see Johnson and Samworth\cite{MR2172843}; and Dorea and Oliveira \cite{MR3185561}).
To define it, let $(\mathcal{X},d)$ be a complete metric space 
and let $\mathcal{P}(\mathcal{X})$ be the set of all probability measures $\mu$ on the 
Borel $\sigma$-field of $\mathcal{X}$. 
The Wasserstein distance of order $p>0$ between two
probability measures $\mu_1, \mu_2\in \mathcal{P}(\mathcal{X})$ is defined as	
\[
d_{p}(\mu_1,\mu_2)
=
\left\{\inf_{\nu\in \Pi(\mu_1,\mu_2) }
\int_{\mathcal{X}\times \mathcal{X}}d^p(x,y)\, \nu(dx,dy)  \right\}^{1/p},
\]	
where $\Pi(\mu_1,\mu_2)$ is the set of all Borel probability measures on 
$\mathcal{X}\times \mathcal{X}$ with
marginals $\mu_1$ and $\mu_2$, respectively.

Note that in the case where $\mathcal{X}=\mathcal{B}$ is a finite-dimensional Euclidean space, with its standard norm $\|\cdot\|$,
the Wasserstein distance of order $p>0$ between Borel probability measures $\mu_1$ and $\mu_2$ 
is alternatively given by
\begin{equation}\label{Mallowsdistance}
	d_{p}(\mu_1,\mu_2)
	=
	\inf_{(X,Y)}
	\left\{\mathbb{E}\|X-Y\|^{p}\right\}^{1/{p}},
\end{equation}
where the infimum is taken over all
$\mathcal{B}$-valued random variables (r.v.'s)  $X$ and $Y$, where
$X$ has law $\mu_1$ and $Y$ has law $\mu_2$.
When convenient, we write $d_p(X, Y)$ instead of $d_p(\mu_1,\mu_2)$ and
we shall remark that this particular case is enough for the purpose of this work.

For $p\geqslant 1$, the Wasserstein distance defines a metric on a subspace of $\mathcal{P}(\mathcal{B})$ 
and bears a close connection with weak convergence. Let $\Gamma_{p}(\mathcal{B})$ be the set
of all probability measures $\mu\in\mathcal{P}(\mathcal{B})$ such that
$\int_{\mathcal{B}} \|x\|^{p}\,  d\mu(x)<+\infty$.

\begin{theorem}[Bickel and Freedman  \cite{MR630103}]
	\label{BickelFriedmantheorem}
	Let $p\geqslant 1$, and $\mu$ and $\{\mu_n\}_{n\geq 1}$ in 
	$\Gamma_{p}(\mathcal{B})$. Then, $\lim_{n\to\infty}d_{p}(\mu_n,\mu)=0$ if 
	and only if, for every bounded continuous function $g:\mathcal{B}\to \mathbb{R}$, we have,
	\[
	\lim_{n\to\infty}
	\displaystyle \int_{\mathcal{B}} g(x)\, d\mu_n(x)
	=
	\int_{\mathcal{B}}g(x)\, d\mu(x) 
	\quad\mbox{and}\quad
	\lim_{n\to\infty}
	\displaystyle\int_{\mathcal{B}}\|x\|^{p}\, d\mu_n(x)
	=
	\int_{\mathcal{B}}\|x\|^{p}\, d\mu(x).
	\]
\end{theorem}

Assume $X\stackrel{d}{=}F$, $Y\stackrel{d}{=}G$ and $(X,Y)\stackrel{d}{=}H$, 
where $H(x,y)=\min\{F(x), G(y)\}$.
Then the following representation result, known as the representation Theorem, whose proof
can be found in Dorea and Ferreira \cite{MR2863804},
will be helpful to evaluate $d_p (F,G)$ when $\mathcal{B}$ is the real line.
\begin{theorem}
	\label{teorema de representacao Mallows}
	For $p\geqslant 1$ we have
	\begin{align*}
		d_p(F,G)
		&=
		\left\{ \int_{\mathbb{R}^2} |x-y|^p dH(x,y)\right\}^{1/p}
		\eqqcolon \
		\{\mathbb{E}_H|X-Y|^p\}^{1/p}.
	\end{align*}
\end{theorem}
This theorem will be used in our one-dimensional setting.
When moving the discussion to the $m$-dimensional setting,
the following generalization will be required, see Bickel and Freedman \cite{MR630103}.

\begin{theorem}\label{theo-rep-geral}
	Let $X$ and $Y$ be the coordinate functions on 
	$\mathcal{B}\times\mathcal{B}$.
	The infimum of $d_{p}(\mu_1,\mu_2)$ is attained by 
	\[
	\int_{\mathcal{B}\times\mathcal{B}}\|x-y\|^p d\pi(x,y),
	\]
	for some probability $\pi$ on $\mathcal{B}\times\mathcal{B}$ such that 
	$\pi X^{-1}=\mu_1$ and $\pi Y^{-1}=\mu_2$.
\end{theorem}

\section{Convergence in Wasserstein distance}\label{sec-conv}
In this section, we consider one-dimensional and $m$-dimensional 
random processes defined on some underlying complete probability space 
$(\Omega,\mathscr{F},\mathbb{P})$.
In what follows, $\mathbb{N}$ denotes the set of positive integers and for $x\in\mathbb{R}$, $\lfloor x \rfloor$ denotes the greatest integer less than or equal to $x$.
Let $\boldsymbol{X}\coloneqq \{X_i: i\in\mathbb{N}\}$ be a $m$-dimensional stochastic process,
\begin{align}\label{blocks}
	S_n\coloneqq \sum_{i=1}^{n}X_i
	\quad 
	\mbox{and}
	\quad 
	Y_{j,n}\coloneqq \sum\limits_{k=(j-1)\ell_n+1}^{j\ell_n} X_k,
	\quad 
	j=1,\ldots,m_n,
\end{align}
where the first $m_n=\lfloor n/\ell_n \rfloor$  blocks 
have size $\ell_n$ large enough such that
\begin{align}\label{condition:1}
	\lim_{n\to\infty}\ell_n=\infty, \quad \lim_{n\to\infty}{n\over \ell_n}=\infty 
	\quad \text{and} \quad \lim_{n\to\infty} {\ell_n^3\over m_n}= 0.
\end{align}
The condition \eqref{condition:1} is satisfied, for example, if we take $\ell_n=n^{\delta}$ with $\delta<1/4$.

Given $t>0,$ we denote by $Z_{t}(n)$ the stabilized partial sum of the r.v.'s
$X_i$'s. That is,
\begin{align}\label{norm-sum}
	Z_{t}(n) \coloneqq {S_{\lfloor nt\rfloor}\over\sigma\sqrt{n}},
\end{align}
where $\sigma^2\in(0,+\infty)$.
When $t=1$, we  simply write $Z_n$ instead of $Z_{1}(n)$.

\subsection{Dimension one}
\label{Positively associated random processes}
The stochastic process $\boldsymbol{X}$ is said
positively associated (see Esary et al.\cite{MR0217826}) if,
given two coordinate-wise non-decreasing functions 
$f,g:\mathbb{R}^n\to \mathbb{R}$ 
and  $i_1,\ldots,i_n\in\mathbb{N}$, we have
\[
\mathrm{Cov}\big(
f(X_{i_1},\ldots, X_{i_n})
,
g(X_{i_1},\ldots, X_{i_n})
\big)
\geqslant 
0,
\]
provided the covariance exists,
and by a coordinate-wise non-decreasing function 
we mean a function $f$ such that
\[
f(x_1,\ldots,x_n)\leqslant f(y_1,\ldots,y_n),
\]
whenever $x_j\leqslant y_j$ for all $j=1,\ldots, n$.  
\begin{lemma}[Newman and Wright \cite{MR624694}]
	\label{lemma-associatedproperties} 
	Let $\boldsymbol{X}$ be positively associated.	
	If all $X_j$'s have finite second moment, then the characteristic 
	functions $\displaystyle\phi_j(r_j)=\mathbb{E}\exp\{i r_jX_j\}$ and 
	$\phi(r_1,\cdots,r_n)=\mathbb{E}\exp\{i\sum_{j=1}^n r_jX_j\}$  satisfy
	\[
	\Bigg|\phi(r_1,\cdots,r_n)-\prod_{j=1}^{n}\phi_j(r_j)\Bigg|
	\leqslant
	{1\over 2}\sum_{1\leqslant j\neq k\leqslant n}|r_jr_k|\mathrm{Cov}(X_j,X_k).
	\]
\end{lemma}

\bigskip 

For the sake of clarity, if $X\stackrel{d}{=}F$ and $Y\stackrel{d}{=}G$, 
then the Wasserstein distance between their distribution functions $F$ and $G$ 
will be denoted by $d_p(X,Y)$.

\begin{theorem}\label{first-theo-1}
	Let $\boldsymbol{X}$ be a  centered, one-dimensional positively associated and stationary stochastic process.
	Assume that the following conditions are satisfied:
	\begin{align}
		& \lim_{n\to\infty}{1\over n}\mathrm{Var}(S_n)
		= 
		\sigma^2\in(0,+\infty), \label{h1}
		\\[0,2cm]
		& \lim_{n\to\infty} {1\over m_n\ell_n}
		\sum_{j=1}^{m_n}
		\mathrm{Var}(Y_{j,n})
		=
		\sigma^2 \quad \text{and} \quad  \label{h2}
		\\[0,2cm]
		&\mathbb{E}|X_j|^3<C_*<+\infty, \quad\forall j\in\mathbb{N}. \label{h3}
	\end{align}
	Then for each $0<p\leqslant2$,
	\begin{align}\label{conv-1}
		\lim_{n\to\infty}d_p({Z_n},Z)= 0 \quad \text{and} \quad
		\lim_{n\to\infty}\mathbb{E}|Z_n|^p=\mathbb{E}|Z|^p, 
	\end{align}
	where $ Z\sim N(0,1)$.
\end{theorem}
\begin{proof}	
	By Minkowski's inequality and by \eqref{Mallowsdistance}, 
	\begin{align*}
		d_2({Z_n},Z)
		\leqslant
		{\rm E}_n^{1/2}(Z_n)+{\rm E}_n^{1/2}(Z),
	\end{align*}
	where 
	${\rm E}_n(X)
	\coloneqq 
	\mathbb{E}\big({1\over\sigma\sqrt{m_n\ell_n}}S_{m_n\ell_n}-X\big)^2$ 
	for any real-valued random variable $X$.
	Then
	\begin{align}\label{limits-proof}
		\lim_{n\to\infty}d_2({Z_n},Z)=0, \quad \text{whenever} \quad 
		\lim_{n\to\infty}{\rm E}_n(Z_n)=0 \quad \text{and} \quad \lim_{n\to\infty}{\rm E}_n(Z)=0.
	\end{align}
	
	Consider the blocks \eqref{blocks} and assume that the block size $\ell_n$ satisfies
	\eqref{condition:1}.
	We will first prove that $\lim_{n\to\infty}{\rm E}_n(Z_n)=0$. For this, 
	using properties of variance and the positivity of covariances, we have  
	\begin{align*}
		{\rm E}_n(Z_n)
		&=
		{1\over \sigma^2}
		\mathrm{Var}
		\left(
		\sum\limits_{j=1}^{m_n\ell_n}
		\Big({1\over\sqrt{n}}
		- 
		{1\over\sqrt{m_n\ell_n}}\Big)
		X_j	
		+ 
		\sum\limits_{j=m_n\ell_n+1}^{n}
		{1\over\sqrt{n}}
		X_j
		\right) 
		\\[0,2cm]
		&\leqslant
		{2\over\sigma^2}\Big({1\over \sqrt{m_n\ell_n}}-{1\over\sqrt{n}}\Big)^2
		\sum\limits_{i,j=1}^{m_n\ell_n}
		\mathrm{Cov}(X_i,X_j)	
		+
		{2\over n\sigma^2}
		\sum\limits_{i,j=m_n\ell_n+1}^{n}
		\mathrm{Cov}(X_i,X_j).
	\end{align*}
	By inequality 
	$
	\big({1\over\sqrt{m_n\ell_n}}\pm{1\over\sqrt{n}}\big)^2
	\leqslant
	{\ell_n\over nm_n (\sqrt{m_n\ell_n}\mp\sqrt{n})^2},
	$
	the rhs of above inequality is 
	\begin{align*}
		&\leqslant
		{2\over n\sigma^2}\Big({\ell_n\over\sqrt{m_n\ell_n}+\sqrt{n}}\Big)^2
		\
		{1\over m_n\ell_n}
		\sum\limits_{i,j=1}^{m_n\ell_n}\mathrm{Cov}(X_i,X_j) 
		\\[0,2cm]	
		&\hspace*{1,3cm}	
		+
		{2\over\sigma^2}\Big({n-m_n\ell_n\over n}\Big)
		\ 
		{1\over n-m_n\ell_n}
		\sum\limits_{i,j=m_n\ell_n+1}^{n}\mathrm{Cov}(X_i,X_j).
	\end{align*}
	Then, by \eqref{h1} ${\rm E}_n(Z_n)\to 0$ as $n\to\infty$.	
	
	The next step is to prove that $\lim_{n\to\infty}{\rm E}_n(Z)=0$.
	In fact, if
	\[
	I_{n}(t)
	\coloneqq
	\left|
	\mathbb{E}
	\exp\Big\{i{t\over\sqrt{m_n\ell_n}}\sum_{j=1}^{m_n} Y_{j,n}\Big\}
	-
	\prod_{j=1}^{m_n}
	\mathbb{E}
	\exp\Big\{i{t\over\sqrt{m_n\ell_n}} Y_{j,n}\Big\}
	\right|,
	\quad t\in\mathbb{R},
	\]
	by Lemma \ref{lemma-associatedproperties},
	\begin{align*}
		I_{n}(t)
		\leqslant \
		{|t|\over 2m_n\ell_n}
		\sum_{\substack{i,j=1\\ i\neq j}}^{m_n}
		\mathrm{Cov}(Y_{i,n},Y_{j,n})
		= \
		{|t|\over 2}
		\Big(
		{1\over m_n\ell_n}
		\mathrm{Var}\Big(\sum_{j=1}^{m_n\ell_n} X_j \Big)
		-
		{1\over m_n\ell_n}
		\sum_{i=1}^{m_n}\mathrm{Var}(Y_{i,n})
		\Big).
	\end{align*}
	Taking $n\to\infty$ in the above inequality, by \eqref{h1}-\eqref{h2} follows that 
	$
	\lim_{n\to\infty}I_n(t)=0.
	$
	That is, the blocks $Y_{j,n}/\sqrt{m_n\ell_n}$ $j=1,\ldots,m_n$ are asymptotically 
	independent.
	On the other hand, denoting 
	$
	\mathcal{M}_n
	\coloneqq
	\sum_{j=1}^{m_n} Y_{j,n}/\sigma\sqrt{m_n\ell_n}	
	$	
	and using the independence of the blocks, it follows from \eqref{h2} 	that
	\begin{align}\label{pru1}
		\lim_{n\to\infty}\mathrm{Var}(\mathcal{M}_n)
		=
		\lim_{n\to\infty}
		{1\over \sigma^2 m_n\ell_n}
		\sum_{j=1}^{m_n} 
		\mathrm{Var}(Y_{j,n})
		=
		1.
	\end{align}
	For each fixed $n$, by using Minkowski's inequality and \eqref{h3} we get	
	$
	\mathbb{E}|Y_{j,n}|^3
	\leqslant
	\ell_n^3 C_*.
	$
	Since the $Y_{j,n}$'s are zero-mean and independent r.v.'s  
	with finite third moment, it follows from Berry-Esseen's Theorem 
	(see, e.g., Feller \cite{MR0270403}) that
	\begin{align*}
		\sup_{x\in\mathbb{R}}
		\left|
		\mathbb{P}
		\Big(	
		{
			{ 		
				\mathcal{M}_n \over \sqrt{\mathrm{Var}(\mathcal{M}_n)}	
			}
			\leqslant x	
		}
		\Big)	
		-\Phi(x)	
		\right|
		\ \leqslant \
		{6m_n\ell_n^3 C_*\over (\sigma^2m_n\ell_n)^{3/2} \{\mathrm{Var}(\mathcal{M}_n)\}^{3/2}},
		\quad \Phi\stackrel{d}{=}N(0,1).
	\end{align*}
	By \eqref{condition:1} and  \eqref{pru1} the rhs of the above inequality tends to zero, 
	when $n\to\infty$. Therefore,	
	$
	{\mathcal{M}_n/ \sqrt{\mathrm{Var}(\mathcal{M}_n)}}		
	\stackrel{\mathscr{D}}{\rightarrow}
	Z,
	$
	where $\stackrel{\mathscr{D}}{\rightarrow}$ denotes convergence in distribution.
	Then, by \eqref{pru1} and 
	by Slutsky's Theorem, 
	$
	\mathcal{S}_n \stackrel{\mathscr{D}}{\rightarrow}Z
	$	
	and
	$
	\lim_{n\to\infty}
	\mathrm{Var}	
	\left(	
	\mathcal{S}_n
	\right)
	=1
	$
	or equivalently (see Theorem \ref{BickelFriedmantheorem})
	$
	\lim_{n\to\infty}d_2({\mathcal{M}_n},Z)= 0.
	$
	
	From Theorem \ref{teorema de representacao Mallows} follows that there exists a r.v.
	$Z^*\stackrel{d}{=}\Phi$ such that the joint distribution of 
	$(\mathcal{M}_n,Z^*)$ is given by $H(x,y)= \min\{F_{\mathcal{M}_n}(x), \Phi(y)\}$ and
	\[
	d_2(\mathcal{M}_n,Z)
	=
	\mathbb{E}(\mathcal{M}_n - Z^*)^2 \to 0 \quad \text{as} \ n\to\infty
	\ \Leftrightarrow \ {\rm E}_n(Z)\to 0 \quad \text{as} \ n\to\infty.
	\]
	
	Hence, by \eqref{limits-proof} $\lim_{n\to\infty}d_2({Z_n},Z)= 0$.
	Therefore, by Theorem \ref{teorema de representacao Mallows} 
	there exists $\tilde{Z}^{*}\stackrel{d}{=}\Phi$ such that
	\[
	d_2({Z}_n,Z)
	=
	\mathbb{E}({Z}_n - \tilde{Z}^*)^2 \to 0 \quad \text{as} \ n\to\infty.
	\]
	From the definition of Wasserstein distance \eqref{Mallowsdistance} and 
	Lyapunov's inequality we have for $0 < p\leqslant 2$,  
	
	\[
	d_p(Z_n,Z)\leqslant \big\{\mathbb{E}|Z_n - \tilde{Z}^*|^p\big\}^{1/p}
	\to 0
	\quad \text{as} \ n\to\infty.
	\]
\end{proof}
\begin{corollary}\label{coro-new-1}
	Under the conditions of Theorem \ref{first-theo-1}, for each $0<p\leqslant2$ we have
	\begin{align}\label{conv-2}
		\lim_{n\to\infty}
		d_p(Z_{t}(n),B_t)= 0 \quad \text{and} \quad 
		\lim_{n\to\infty}\mathbb{E}|Z_{t}(n)|^p=\mathbb{E}|B_t|^p, \quad  t\in[0,1],
	\end{align}	
	where $B_t$ is the standard one-dimensional Brownian motion.
\end{corollary}
\begin{proof}
	Assume first that $p=2$.
	By Theorem \ref{first-theo-1} we have 
	$\lim_{n\to\infty}d_2(Z_n,Z)= 0$ or equivalently 
	\begin{align}\label{ineq-im}
		Z_n\stackrel{\mathscr{D}}{\rightarrow}Z, \  Z\sim N(0,1),
		\quad 
		\mbox{and}
		\quad 
		\lim_{n\to\infty}\mathrm{Var}(Z_n)= 1.
	\end{align}
	Since,
	\[
	{1\over\sigma\sqrt{n}}S_{\lfloor nt\rfloor}
	=
	\sqrt{\lfloor nt\rfloor\over n} \ {1\over\sigma\sqrt{{\lfloor nt\rfloor}}}S_{\lfloor nt\rfloor}
	\quad \text{and} \quad 
	\lim_{n\to\infty}{\lfloor nt\rfloor\over n}= t,
	\]
	by \eqref{ineq-im} and Slutsky's Theorem we obtain
	$
	Z_t(n)\stackrel{\mathscr{D}}{\rightarrow} B_t
	$
	and
	$
	\lim_{n\to\infty}\mathrm{Var}(Z_t(n))=t.
	$
	By Theorem \ref{BickelFriedmantheorem}, it follows that
	$\lim_{n\to\infty}d_2(Z_{t}(n),B_t)= 0$.
	
	As in the proof of Theorem \ref{first-theo-1}, 
	the Lyapunov  inequality completes the proof for $0 <p < 2$.
\end{proof}

\subsection*{Applications of Theorem \ref{first-theo-1}}

The examples and discussion presented in this section are inspired by the ones in 
the preprint Cioletti et al. \cite{2017arXiv170103747C}.

For the volume $\Lambda_N = \{1,2,\dots ,N\}$, the Gibbs measure of 
the one-dimensional Ising model  with free boundary conditions, at inverse 
temperature $\beta>0$ on $\Lambda_N$, is given by
\begin{align}\label{Ising model}
	\mu_{N}^{\beta}(\sigma)
	=
	{1\over		Z_N(\beta)
	}
	\exp\Big(
	\beta 
	\sum_{1\le i<j\le N}V_{ij}\sigma_i\sigma_j
	\Big),
\end{align}
where $\sigma=(\sigma_1,\ldots,\sigma_N)\in \{-1,1\}^N$, the interactions
$V_{ij}$  being the same we used
to define our two-dimensional random polymer model \eqref{hamiltonian}
and
\[
Z_N(\beta)=\sum_{\sigma}\exp\Big(
\beta 
\sum_{1\leqslant i<j\leqslant N}V_{ij}\sigma_i\sigma_j
\Big)
\]
is the normalization factor.

Let $\mu^{\beta}$ be the thermodynamical limit of $\mu_{\Lambda}^{\beta}$ when 
$\Lambda\to \mathbb{N}$.
The susceptibility $\chi(\beta)$ of the Ising model at the inverse temperature $\beta$ 
is defined by
\begin{align}\label{susceptibility}
	\chi(\beta)
	\coloneqq
	\sum_{n=1}^{\infty}
	\left\{
	\mu^{\beta}(\sigma_1\sigma_n)-\mu^{\beta}(\sigma_1)\mu^{\beta}(\sigma_n)
	\right\}.
\end{align}

Let $\boldsymbol{X}=\{X_i:i\in\mathbb{N}\}$ be a stochastic process
defined on $\Omega=\{-1,1\}^{\mathbb{N}}$, where the variables $X_i$
are projections, i.e., for
$\sigma=(\sigma_{1},\sigma_{2},\ldots)\in\{-1,1\}^{\mathbb{N}}$ we 
have $X_i(\sigma)=\sigma_{i}, \, \forall i.$
\begin{example}\label{ex-1}
	For fixed $L> 0$, define $V_{ij}$ as: 
	\[
	V_{ij}=V,\quad i,j\in\mathbb{N} \ \text{s.t.} \ 0<|i-j|\leqslant L,
	\]
	where $V>0$ is a constant.
	In this case, it is well-known that the
	set of the Gibbs measures $\mathscr{G}(\beta)$
	is a singleton. 
	It can be verified that $\boldsymbol{X}$ on $(\Omega,\mathscr{F},\mu)$ is not a sequence of
	independent r.v.'s by applying the GKS-II inequality. 
	Moreover, one can verify that $\boldsymbol{X}$ on $(\Omega,\mathscr{F},\mu)$
	is stationary and positively
	associated by using the FKG inequality; see Fortuin et al. \cite{MR0309498}. 
	From the Lieb-Simon inequality $($cf. Lieb \cite{MR589427} and Simon \cite{MR589426}$)$
	follows that the susceptibility $\chi(\beta)<+\infty$
	$($for a more recent version see Duminil-Copin and Tassion \cite{MR3477351}$)$. 
	By a straightforward calculation, we have	
	\begin{align*}
		\lim_{n\to\infty}
		{1\over n}\mathrm{Var}(S_n)
		= 
		\mathrm{Var}(\sigma_1)
		+ 
		2\lim_{n\to\infty}
		{1\over n} \sum_{j=2}^n \sum_{i=1}^{j-1} \mathrm{Cov}(\sigma_1,\sigma_j)
		= 
		\chi(\beta)
	\end{align*}
	and	similarly
	$
	\lim_{n\to\infty}	
	\sum_{j=1}^{m_n}\mathrm{Var}(Y_{j,n})/m_n\ell_n
	=	
	\chi(\beta).
	$
	Therefore, all the hypotheses of Theorem
	\ref{first-theo-1} hold. Moreover, the convergences 
	\eqref{conv-1}  and \eqref{conv-2} also hold for $0<p\leqslant 2$.
\end{example}
\begin{example}\label{exemplo-ising-longo-alcance}
	For all $i\in\mathbb{N}$ we define $V_{ii}=0$ and 
	\[
	V_{ij} = \beta |i-j|^{-\alpha}, \quad i,j\in\mathbb{N}\ \text{and}\ i\neq j
	\]
	where $\beta>0$ and $\alpha>1$.
	The analysis in terms of the parameter $\alpha$ is twofold.
	The first one is $1<\alpha\leqslant 2$ and the 
	second is $\alpha>2$. 
	
	Suppose that $1<\alpha\leqslant 2$. 
	In this case there is a real number $\beta_c(\alpha)\in(0,+\infty)$ called critical point such that, 
	for all $\beta<\beta_c(\alpha)$ the set of the Gibbs measures $\mathscr{G}(\beta)$
	is a singleton 
	$($see Dyson \cite{MR0436850}; and Fr$\rm\ddot{o}$hlich and Spencer \cite{MR660541}$)$ and 
	the unique probability measure $\mu_{\beta,\alpha}$  
	has the FKG property and the 
	stochastic process $\boldsymbol{X}= \{X_i:i\in\mathbb{N}\}$ on 
	$(\Omega,\mathscr{F},\mu_{\beta,\alpha})$ is associated
	and stationary. 
	At high temperature the covariance $\mathrm{Cov}_{\mu_{\beta,\alpha}}(X_0,X_i)$ decays polynomially
	with the same rate as the interaction decay, and since $\alpha>1$, in particular
	the susceptibility $\chi(\mu_{\beta,\alpha})$ is finite.
	The conditions \eqref{h1}-\eqref{h3} are checked in analogy to that made in Example \ref{ex-1}. 
	In this case, the convergences \eqref{conv-1} and \eqref{conv-2} hold for $0<p\leqslant 2$.	
	On the other hand, for all $\beta>\beta_c(\alpha)$,
	the set $\mathscr{G}(\beta)$ 
	has infinitely many elements. Then, we can not
	ensure that the stochastic process $\boldsymbol{X}$ on 
	$(\Omega,\mathscr{F},\mu)$ is stationary
	for any $\mu\in\mathscr{G}(\beta)$. Moreover,
	the susceptibility is not finite anymore. 	
	
	The case $\alpha>2$ is similar to the case $1<\alpha\leqslant 2$
	and $\beta<\beta_{c}(\alpha)$, but no restriction on the parameter $\beta$
	is needed to ensure the uniqueness of the Gibbs measures and the other
	used properties.	
\end{example}

\subsection{Higher dimensions}
\label{2-Positively associated random processes}
Given $m\in\mathbb{N}$, 
let $\{(X^1_i,\ldots,X^m_i):i\in\mathbb{N}\}$ be a $m$-dimensional stationary random process 
with independent coordinates. We define
\begin{align}\label{norm-sum-1}
	S_n^j
	\ \coloneqq \
	\sum_{i=1}^n X^j_i,
	\qquad 
	Z^j_{t}(n) \ \coloneqq \ {S^j_{\lfloor nt\rfloor}\over\sigma\sqrt{n}},
	\quad j=1,\ldots,m,
\end{align}
$\boldsymbol{S}_{n}^m\coloneqq(S_n^1,\ldots, S_n^m)$ and $\boldsymbol{Z}_{n}^m(t)\coloneqq(Z^1_{t}(n),\ldots, Z^m_{t}(n))$,
where $\sigma^2\in(0,+\infty)$.
\begin{theorem}\label{second-theo}
	Let $\{X^j_i:i\in\mathbb{N}\}$, $j=1,\ldots,m$ be a centered, one-dimensional 
	positively associated and stationary random process satisfying the hypotheses 
	of Theorem \ref{first-theo-1}. Then, for each $0<p\leqslant 2$ we have	
	\[
	\lim_{n\to\infty}	
	d_p(\boldsymbol{Z}_{n}^m(t),\boldsymbol{B}_m(t))
	=0
	\quad \text{and} \quad 
	\lim_{n\to\infty}
	\mathbb{E}\|\boldsymbol{Z}_{n}^m(t)\|^p=\mathbb{E}\|\boldsymbol{B}_m(t)\|^p,
	\]
	where $\boldsymbol{B}_m(t)=(B_t^1,\ldots, B_t^m)$ is the
	$m$-dimensional standard Brownian motion.
\end{theorem}
\begin{proof} In view of Theorem \ref{teorema de representacao Mallows}, there exists an r.v.
	$B^{j,*}_{t}\stackrel{d}{=}B^j_{t}$ such that the joint distribution of 
	$(Z^j_{t}(n),B^{j,*}_{t})$ is given by $H(x,y)=\min\{F_{Z^j_{t}(n)}(x),F_{B^{j,*}_{t}}(y)\}$ and
	\[
	d_2(Z^j_{t}(n),B^{j,*}_{t})
	=
	\mathbb{E}(Z^j_{t}(n) - B^{j,*}_{t})^2,
	\quad j=1,\ldots,m.
	\]
	By the above identity and by Minkowski's inequality,
	\begin{align}\label{ref-imp}	
		d_2(\boldsymbol{Z}_{n}^m(t), \boldsymbol{B}_m(t))
		&\ \leqslant \
		\big\{\mathbb{E}\big\|\boldsymbol{Z}_{n}^m(t) - \boldsymbol{B}_m(t) \big\|^2 \big\}^{1/2}
		\nonumber	
		\\[0,2cm]
		&\ \leqslant \
		\sum_{j=1}^{m}\big\{\mathbb{E}(Z^j_{t}(n) - B_t^j )^2\big\}^{1/2}
		=	
		\sum_{j=1}^{m} d_2(Z^j_{t}(n),B_t^j).
	\end{align}
	Combining this inequality with \eqref{conv-2}, we obtain the convergence of order
	$p=2$.
	
	As in the proof of Theorem \ref{first-theo-1}, 
	the Lyapunov inequality completes the proof for $0 <p < 2$.	
\end{proof}
\begin{theorem}\label{corollary1}
	Under the conditions of Theorem \ref{second-theo}, consider the random process	
	\[	
	{W}_{n,m}(t)\coloneqq {1\over\sigma\sqrt{2n}}
	\left\{\boldsymbol{S}_{k}^m+(nt-k) (\boldsymbol{S}_{k+1}^m-\boldsymbol{S}_{k}^m)\right\},	
	\quad 
	{k\over n}\leqslant t<{k+1\over n},	\ t\in[0,1].
	\]
	Then, for each $0<p\leqslant 2$ we have	
	\[	
	\lim_{n\to\infty}
	d_p({W}_{n,m}(t),\boldsymbol{B}_m(t))= 0
	\quad \text{and} \quad 
	\lim_{n\to\infty}
	\mathbb{E}\|{W}_{n,m}(t)\|^p=\mathbb{E}\|\boldsymbol{B}_m(t)\|^p,
	\]
	where $\boldsymbol{B}_m(t)=(B_t^1,\ldots, B_t^m)$ is the
	$m$-dimensional standard Brownian motion.
\end{theorem}
\begin{proof}	
	The inequality ${k\over n}\leqslant t<{k+1\over n}$ implies that $k=\lfloor nt\rfloor$. 
	Then,	
	\begin{align*}
		{W}_{n,m}(t)= 
		{1\over\sqrt{2}}	
		\left\{
		\boldsymbol{Z}_{n}^m(t)
		+
		\frac{(nt-{\lfloor nt\rfloor})}{\sigma\sqrt{n}}
		(\boldsymbol{S}_{\lfloor nt\rfloor+1}^m-\boldsymbol{S}_{\lfloor nt\rfloor}^m)
		\right\}
		\eqqcolon	
		{1\over\sqrt{2}}	
		\left\{
		\boldsymbol{Z}_{n}^m(t)
		+
		\boldsymbol{\psi}^m_n(t)
		\right\}.
	\end{align*}	
	Note that
	\begin{align}\label{ig-esp}
		\big\|\boldsymbol{\psi}_n^m(t) \big\|^2
		=
		\frac{(nt-{\lfloor nt\rfloor})}{\sigma\sqrt{n}}
		\sum_{j=1}^{m} (X_{k+1}^j)^2.
	\end{align}
	By the above decomposition of ${W}_{n,m}(t)$ and Minkowski's inequality,
	\begin{align*}	
		d_2({W}_{n,m}(t),\boldsymbol{B}_m(t))
		&\leqslant
		{1\over \sqrt{2}}
		\big\{
		\mathbb{E}\|{W}_{n,m}(t)-\boldsymbol{B}_m(t) \|^2
		\big\}^{1/2}
		\\[0,2cm]	
		&\leqslant	
		{1\over \sqrt{2}}
		\big\{
		\mathbb{E}\|\boldsymbol{Z}_{n}^m(t) -\boldsymbol{B}_m(t)\|^2
		\big\}^{1/2}
		+
		{1\over \sqrt{2}}
		\big\{
		\mathbb{E}
		\|\boldsymbol{\psi}_n^m(t) \|^2
		\big\}^{1/2}
		\\[0,2cm]
		&\leqslant
		{1\over \sqrt{2}}
		\sum_{j=1}^{m} d_2(Z^j_{t}(n),B_t^j)
		+
		{1\over \sqrt{2}}
		\left\{
		\frac{(nt-{\lfloor nt\rfloor})}{\sigma\sqrt{n}}
		\sum_{j=1}^{m} \text{Var}(X_{k+1}^j)
		\right\}^{1/2},
	\end{align*}
	where in the last inequality we used \eqref{ref-imp} and \eqref{ig-esp}.
	Taking $n\to\infty$ in the above inequality and
	using Corollary \ref{coro-new-1}, we get
	$\lim_{n\to\infty}d_2({W}_{n,m}(t),\boldsymbol{B}_m(t))=0$.				
\end{proof}
The main result of this paper is an application
of the results of this section for a two-dimensional random
polymer model presented in the next section.

\section{Proof of Theorem \ref{Donsker invariance principle}}\label{sec-proof}
In this section,
we consider the random polymer model on $\mathbb{Z}^2$ given by the 
Gibbs measure \eqref{model}-\eqref{hamiltonian}, where the interaction $V_{ij}=|i-j|^{-\alpha}$,
and $1<\alpha\leqslant 2$. 

It is very well-known (see Dyson \cite{MR0436850}) 
that below the critical temperature 
the one-dimensional Ising 
model with such interactions exhibits a first-order phase transition in the magnetic field and that
at the critical temperature 
there is either a second-order or a mixed first-order-second-order 
(Thouless effect) transition (see Aizenman et al. \cite{Aizenman1988}).
We will prove that below this critical value $\beta_c$ (respectively, above the critical inverse temperature),
the random polymer model, after a rescaling, converges (respectively, does not converge)  
in Wasserstein distance to the planar standard Brownian motion,
leaving open the problem of asymptotic behavior of the random polymers in the critical phase. 

\subsection{Proof of Item 1.}
\noindent
Let $T:\mathbb{R}^2\to\mathbb{R}^2$ be the rotation through an angle $\pi/4$.
If $\mathcal{S}=(\mathcal{S}_1,\ldots,\mathcal{S}_n)\in\mathbb{W}_n$ and 
$\mathcal{X}_i=\mathcal{S}_{i+1}-\mathcal{S}_i$, then for each step $\mathcal{X}_i$	
there are $\sigma_i^{(1)},\sigma_i^{(2)}\in \{-1,1\}$ so that	
\begin{align}\label{definicao-sigmas-via-T}
	T\mathcal{X}_i = \sigma_i^{(1)} \frac{e_1}{\sqrt{2}}+\sigma_i^{(2)}\frac{e_2}{\sqrt{2}}.
\end{align}

Let $\mu^\beta$ be the thermodynamical limit of $\mu^\beta_\Lambda$ when $\Lambda\to\mathbb{N}$, 
and $\chi(\beta)$ the susceptibility of the Ising model at the inverse temperature $\beta$
defined in \eqref{susceptibility}.
Since $V_{ij}\geqslant 0$, it follows from the FKG inequality \cite{MR0309498} that 
$\chi(\beta)\geqslant 0$ and that the sequence of r.v.'s 
$\{\sigma_i^{(j)}:i\in\mathbb{N}\}$, $j=1,2$, are associated with respect to the 
infinite-volume Gibbs measure $\mu^\beta$.

Let $S_k^j=\sum_{i=1}^{k}X_i^j$ with $X_i^j\coloneqq\sigma_i^{(j)}$, $j=1,2$.
From \eqref{definicao-sigmas-via-T} follows that $T\mathcal{W}_n(t)$ 
has the following expression 
\begin{align*}
	T\mathcal{W}_n(t)
	&= 
	{1\over\sqrt{2n\chi(\beta)}}
	\sum_{j=1}^{2}
	\left(S_{\lfloor nt\rfloor}^j+(nt-\lfloor nt\rfloor) 
	(S_{\lfloor nt\rfloor +1}^j-S_{\lfloor nt\rfloor}^j) \right) e_j
	\\[0,2cm]
	&= 
	{W}_{n,2}(t),
\end{align*}
where ${W}_{n,2}(t)$ is the random process defined in Theorem \ref{corollary1}
with $m=2$ and $e_1, e_2$ are elements of the canonical basis of $\mathbb{R}^2$.
From Aizenman-Barsky-Fern\'andez's Theorem 
(Aizenman et al. \cite{MR894398}, see also Duminil-Copin and Tassion \cite{MR3477351}) 
follows that the susceptibility is finite as long as $\beta<\beta_c$. 
Let $\mathscr{G}(\beta)$ be the set of all infinite-volume
Gibbs measures. It is well-known that 
the set $\mathscr{G}(\beta)=\{\mu^\beta\}$
is a singleton for any $\beta<\beta_c$. 
In this case, this unique measure $\mu^\beta$ is 
translation invariant and therefore the r.v's 
$\{X_i^{j}:i\in\mathbb{N}\}$, $j=1,2$, form a
stationary sequence. Since these r.v.'s are uniformly bounded, they
have finite third moment and so \eqref{h3} is valid.
The conditions \eqref{h1} and \eqref{h2} of Theorem \ref{first-theo-1}
are proved analogously to Example \ref{ex-1}, so 
all the hypotheses of Theorem \ref{corollary1} hold. Therefore,
\begin{align}\label{con-zero-m}
	\lim_{n\to\infty}
	d_p({W}_{n,2}(t),\boldsymbol{B}_2(t))= 0.
\end{align}

On the other hand,
since $T$ is invertible and its inverse is a
linear transformation, denoted by $T^{-1}$, we have $\mathcal{W}_n(t)=T^{-1}{W}_{n,2}(t)$.
Combining this equality with Item (8.3) of reference Mallows \cite{MR0298812} and with the fact that the
distribution of the Brownian motion is invariant under rotations
in the plane, we obtain
\[
d_2(\mathcal{W}_n(t),\boldsymbol{B}_2(t))
=
d_2\big(T^{-1}{W}_{n,2}(t), \boldsymbol{B}_2(t) \big)
\leqslant
\|T^{-1}\| \cdot d_2\big({W}_{n,2}(t), \boldsymbol{B}_2(t) \big).
\]
By taking $n\to\infty$, in the above inequality, we get from \eqref{con-zero-m} that
$d_2(\mathcal{W}_n(t),\boldsymbol{B}_2(t))\to 0$.

Again, by applying the Lyapunov's inequality 
we complete the proof for $0 <p < 2$. \qed

\subsection{Proof of Item 2.}
\noindent
By Theorem \ref{BickelFriedmantheorem} it is sufficient to prove that $\mathcal{W}_n(t)\stackrel{\mathscr{D}}{\nrightarrow} \boldsymbol{B}_2(t)$ for each $\beta>\beta_c$.
The proof is by contradiction. Suppose that 
$\mathcal{W}_n(t)\stackrel{\mathscr{D}}{\rightarrow} \boldsymbol{B}_2(t)$.
So for each $t\in\mathbb{R}$, we have
\begin{align}\label{prim1}
	\lim_{n\to\infty}\mathbb{E}\|\mathcal{W}_n(t)\|_1^2=\mathbb{E}\|\boldsymbol{B}_2(t)\|_1^2=2t.
\end{align}
From now, to simplify the notation we write $\mathcal{W}_n(t)$ as
\[
\mathcal{W}_n(t)
=
\left\{
{1\over \sigma\sqrt{n}}
\mathcal{S}_{\lfloor nt\rfloor}
+
{(nt-\lfloor nt\rfloor)\over \sigma\sqrt{n}}
(\mathcal{S}_{\lfloor nt\rfloor +1}-\mathcal{S}_{\lfloor nt\rfloor})
\right\}
\eqqcolon
{1\over \sigma\sqrt{n}}\mathcal{S}_{\lfloor nt\rfloor}
+
\psi_n(t).
\]
For any $t>0$, the mean square of $\mathcal{W}_n(t)$ is given by
\begin{align*}
	\mathbb{E}\|\mathcal{W}_n(t)\|_1^2
	&=
	{1\over n\sigma^2}\mathbb{E}\|\mathcal{S}_{\lfloor nt\rfloor}\|_1^2
	+
	{1\over \sigma\sqrt{n}}
	\mathbb{E}\langle \mathcal{S}_{\lfloor nt\rfloor}, \psi_n(t)\rangle + \mathbb{E}\|\psi_n(t)\|_1^2
	\\[0,2cm]
	&=
	{1\over n\sigma^2}\mathbb{E}\|\mathcal{S}_{\lfloor nt\rfloor}\|_1^2
	+
	{\sqrt{nt-\lfloor nt\rfloor}\over \sigma\sqrt{n}}
	\mathbb{E}
	\Bigg\langle{\mathcal{S}_{\lfloor nt\rfloor}\over\sqrt{\lfloor nt\rfloor}},\psi_n(t) \Bigg\rangle
	+
	\mathbb{E}\|\psi_n(t)\|_1^2
	\\[0,2cm]
	&\eqqcolon
	a_n(t)+b_n(t)+c_n(t).
\end{align*}
It is simple to check that 
$
\psi_n(t)=
{(nt-\lfloor nt\rfloor)\over \sigma\sqrt{n}}
(\mathcal{S}_{\lfloor nt\rfloor +1}-\mathcal{S}_{\lfloor nt\rfloor})
\stackrel{\mathscr{D}}{\rightarrow} 0
$
and $\lim_{n\to\infty}c_n(t)=0$. We also have $\lim_{n\to\infty}b_n(t)=0$, since
\[
{\mathcal{S}_{\lfloor nt\rfloor}\over \sqrt{\lfloor nt\rfloor}}
\stackrel{\mathscr{D}}{\rightarrow} \boldsymbol{B}_2(t)
\quad \text{and} \quad 
{\sqrt{nt-\lfloor nt\rfloor}\over \sigma\sqrt{n} }\to 0.
\]
Finally, applying the main theorem of Cioletti et al. \cite{MR3226841} on the thermodynamic 
limit, we get
\[
{1\over 2}m_*(\beta)^2
\leqslant
{1\over \lfloor nt\rfloor^2}\mathbb{E}\|\mathcal{S}_{\lfloor nt\rfloor}\|_1^2
\leqslant 1.
\]
This shows that for any $\varepsilon>0$ 
there exists $N_0$ so that, if $n\geqslant N_0$, then
\[
{\lfloor nt\rfloor^2 \over 2n\sigma^2}m_*(\beta)^2
\leqslant
a_n(t)
\leqslant
{\lfloor nt\rfloor^2 \over n\sigma^2}+\varepsilon,
\]
which implies $\lim_{n\to\infty}a_n(t)=+\infty$. Hence, for each $t>0$,
\[
\lim_{n\to\infty}\mathbb{E}\|\mathcal{W}_n(t)\|_1^2=+\infty,
\]
which contradicts \eqref{prim1}, thus finishing the proof. \qed


\section*{Acknowledgements}

This study was financed in part by the Coordena\c{c}\~{a}o de Aperfei\c{c}oamento de Pessoal de N\'{i}vel Superior - Brasil (CAPES) - Finance Code 001.
The authors are grateful to the anonymous referee for relevant comments and
suggestions.
It is a pleasure to thank Leandro Cioletti for many valuable comments and
careful reading of this manuscript.

%
%
%

\begin{thebibliography}{10}
	
	\bibitem{MR894398}
	M.~Aizenman, D.~J. Barsky, and R.~Fern\'andez, \emph{The phase transition in a
		general class of {I}sing-type models is sharp}, J. Statist. Phys. \textbf{47}
	(1987), no.~3-4, 343--374. \MR{894398}
	
	\bibitem{Aizenman1988}
	M.~Aizenman, J.~T. Chayes, L.~Chayes, and C.~M." Newman, \emph{Discontinuity of
		the magnetization in one-dimensional $1/|x- y|^2$ ising and potts models},
	Journal of Statistical Physics \textbf{50} (1988), no.~1, 1--40.
	
	\bibitem{MR2011032}
	Luigi Ambrosio, \emph{Lecture notes on optimal transport problems},
	Mathematical aspects of evolving interfaces ({F}unchal, 2000), Lecture Notes
	in Math., vol. 1812, Springer, Berlin, 2003, pp.~1--52. \MR{2011032}
	
	\bibitem{MR630103}
	P.~J. Bickel and D.~A. Freedman, \emph{Some asymptotic theory for the
		bootstrap}, Ann. Statist. \textbf{9} (1981), no.~6, 1196--1217. \MR{630103}
	
	\bibitem{MR3456987}
	Jean-P. Bouchaud, M.~Mezard, G.~Parisi, and J.~S. Yedidia, \emph{Polymers with
		long-range self-repulsion: a variational approach}, Journal of Physics A:
	Mathematical and General \textbf{24} (1991), no.~17, L1025--L1030.
	
	\bibitem{MR2151217}
	Paolo Butt\`a, Aldo Procacci, and Benedetto Scoppola, \emph{Kac polymers}, J.
	Stat. Phys. \textbf{119} (2005), no.~3-4, 643--658. \MR{2151217}
	
	\bibitem{MR1301458}
	Sergio Caracciolo, Giorgio Parisi, and Andrea Pelissetto, \emph{Random walks
		with short-range interaction and mean-field behavior}, J. Statist. Phys.
	\textbf{77} (1994), no.~3-4, 519--543. \MR{1301458}
	
	\bibitem{MR3226841}
	L.~Cioletti, C.~C.~Y. Dorea, and S.~Vasconcelos da~Silva,
	\emph{Diffusive-ballistic transition in random polymers with drift and
		repulsive long-range interactions}, J. Stat. Phys. \textbf{156} (2014),
	no.~4, 760--765. \MR{3226841}
	
	\bibitem{2017arXiv170103747C}
	L.~{Cioletti}, C.~C.~Y. {Dorea}, and R.~{Vila}, \emph{{Limit Theorems in
			Mallows Distance for Processes with Gibbsian Dependence}}, ArXiv-1701.03747
	(2017).
	
	\bibitem{MR2863804}
	C.~C.~Y. Dorea and D.~B. Ferreira, \emph{Conditions for equivalence between
		{M}allows distance and convergence to stable laws}, Acta Math. Hungar.
	\textbf{134} (2012), no.~1-2, 1--11. \MR{2863804}
	
	\bibitem{MR3185561}
	C.~C.~Y. Dorea and M.~A. Oliveira, \emph{The {D}onsker's theorem for {L}\'evy
		stable motions via {M}allows distance}, Markov Process. Related Fields
	\textbf{20} (2014), no.~1, 167--172. \MR{3185561}
	
	\bibitem{MR3477351}
	Hugo Duminil-Copin and Vincent Tassion, \emph{A new proof of the sharpness of
		the phase transition for {B}ernoulli percolation and the {I}sing model},
	Comm. Math. Phys. \textbf{343} (2016), no.~2, 725--745. \MR{3477351}
	
	\bibitem{MR0436850}
	F.~J. Dyson, \emph{Existence of a phase-transition in a one-dimensional {I}sing
		ferromagnet}, Comm. Math. Phys. \textbf{12} (1969), no.~2, 91--107.
	\MR{0436850}
	
	\bibitem{MR0217826}
	J.~D. Esary, F.~Proschan, and D.~W. Walkup, \emph{Association of random
		variables, with applications}, Ann. Math. Statist. \textbf{38} (1967),
	1466--1474. \MR{0217826}
	
	\bibitem{MR0270403}
	W.~Feller, \emph{An introduction to probability theory and its applications.
		{V}ol. {II}}, Second edition, John Wiley \& Sons, Inc., New
	York-London-Sydney, 1971. \MR{0270403}
	
	\bibitem{MR0309498}
	C.~M. Fortuin, P.~W. Kasteleyn, and J.~Ginibre, \emph{Correlation inequalities
		on some partially ordered sets}, Comm. Math. Phys. \textbf{22} (1971),
	89--103. \MR{0309498}
	
	\bibitem{MR660541}
	J.~Fr\"ohlich and T.~Spencer, \emph{The phase transition in the one-dimensional
		{I}sing model with {$1/r^{2}$}\ interaction energy}, Comm. Math. Phys.
	\textbf{84} (1982), no.~1, 87--101. \MR{660541}
	
	\bibitem{MR2172843}
	O.~Johnson and R.~Samworth, \emph{Central limit theorem and convergence to
		stable laws in {M}allows distance}, Bernoulli \textbf{11} (2005), no.~5,
	829--845. \MR{2172843}
	
	\bibitem{MR0102006}
	L.~V. Kantorovi\v{c} and G.~\v{S}. Rubin\v{s}te\u{\i}n, \emph{On a space of
		completely additive functions}, Vestnik Leningrad. Univ. \textbf{13} (1958),
	no.~7, 52--59. \MR{0102006}
	
	\bibitem{MR589427}
	E.~H. Lieb, \emph{A refinement of {S}imon's correlation inequality}, Comm.
	Math. Phys. \textbf{77} (1980), no.~2, 127--135. \MR{589427}
	
	\bibitem{MR0298812}
	C.~L. Mallows, \emph{A note on asymptotic joint normality}, Ann. Math. Statist.
	\textbf{43} (1972), 508--515. \MR{0298812}
	
	\bibitem{MR1234567}
	E.~Marinari and G.~Parisi, \emph{Aon polymers with long range repulsive
		forces}, Europhys Lett. \textbf{15} (1991), 721--724.
	
	\bibitem{MR624694}
	C.~M. Newman and A.~L. Wright, \emph{An invariance principle for certain
		dependent sequences}, Ann. Probab. \textbf{9} (1981), no.~4, 671--675.
	\MR{624694}
	
	\bibitem{MR2379704}
	Aldo Procacci, R\'emy Sanchis, and Benedetto Scoppola,
	\emph{Diffusive-ballistic transition in random walks with long-range
		self-repulsion}, Lett. Math. Phys. \textbf{83} (2008), no.~2, 181--187.
	\MR{2379704}
	
	\bibitem{MR589426}
	B.~Simon, \emph{Correlation inequalities and the decay of correlations in
		ferromagnets}, Comm. Math. Phys. \textbf{77} (1980), no.~2, 111--126.
	\MR{589426}
	
	\bibitem{MR1633582}
	Remco van~der Hofstad, Frank den Hollander, and Gordon Slade, \emph{A new
		inductive approach to the lace expansion for self-avoiding walks}, Probab.
	Theory Related Fields \textbf{111} (1998), no.~2, 253--286. \MR{1633582}
	
	\bibitem{MR0314115}
	L.~N. Vasershtein, \emph{Markov processes over denumerable products of spaces
		describing large system of automata}, Problemy Pereda\v ci Informacii
	\textbf{5} (1969), no.~3, 64--72. \MR{0314115}
	
\end{thebibliography}

\providecommand{\bysame}{\leavevmode\hbox to3em{\hrulefill}\thinspace}
\providecommand{\MR}{\relax\ifhmode\unskip\space\fi MR }
\providecommand{\MRhref}[2]{%
	\href{http://www.ams.org/mathscinet-getitem?mr=#1}{#2}
}
\providecommand{\href}[2]{#2}

\end{document}